\documentclass[12pt]{amsart}
\usepackage{amsmath,amssymb,amsbsy,amsfonts,amsthm,latexsym,
                        amsopn,amstext,amsxtra,euscript,amscd,mathrsfs,color}

\usepackage{float}
\usepackage[english]{babel}
\usepackage{url}
\usepackage[colorlinks,linkcolor=blue,anchorcolor=blue,citecolor=blue]{hyperref}

\newtheorem{theorem}{Theorem}
\newtheorem{lemma}[theorem]{Lemma}
\newtheorem{example}[theorem]{Example}

\newtheorem{remark}[theorem]{Remark}

\def\F{\mathbb{F}}

\def\Tr{\text{\rm Tr}}
\def\wt{\text{\rm wt}}
\def\Ker{\text{\rm Ker}}

\def\Supp{\text{\rm Supp}}
\def\Im{\text{\rm Im}}
\def\d{\textrm{d}}
\def\x{\textrm{x}}
\def\y{\textrm{y}}
\def\D{\textrm{D}}
\def\c{\textrm{c}}
\def\Prj{\textrm{Prj}}

\numberwithin{equation}{section}
\numberwithin{theorem}{section}

\begin{document}

\title[Weight distributions and weight hierarchies]{Weight distributions and weight hierarchies of a family of $p$-ary linear codes}
\markright{Weight hierarchies and weight distributions}
\author{Fei Li}
\address{Faculty of School of Statistics and Applied Mathematics,
Anhui University of Finance and Economics, Bengbu,  Anhui Province, 233030, China}
\email{cczxlf@163.com; 120110029@aufe.edu.cn}

\author{Xiumei Li}
\address{School of Mathematical Sciences, Qufu Normal University, Qufu Shandong, 273165, China}
\email{lxiumei2013@qfnu.edu.cn}

\begin{abstract}
The weight distribution and weight hierarchy of linear codes are two important research topics in coding theory. In this paper, by choosing proper defining sets from inhomogeneous quadratic functions over $\mathbb{F}_{q}^{2},$
we construct a family of $3$-weight $p$-ary linear codes and
determine their weight distributions and weight hierarchies.
Most of the codes can be used in secret sharing schemes.
\end{abstract}

\maketitle

\section{Introduction}
\label{intro}


For an odd prime number $p$ and a positive integer $e$, let $ \mathbb{F}_{q} $ be the finite field with $ q=p^{e} $
elements and $\mathbb{F}_{q}^{*}$ be its multiplicative group.

An $[n,k,d]$ $p$-ary linear code $C$ is a $k$-dimensional subspace of $ \mathbb{F}_{p}^{n} $ with minimum (Hamming) distance $d$.
For $0\leq i \leq n$, let $A_{i}=|\{\x\in C : \wt(\x)=i\}|$. Obviously, $A_{0}=1$.
The weight distribution of  $C$ is defined as the sequence $1,A_{1},\cdots,A_{n}$.
A code $C$ is said to be $t$-weight if $t=|\{i : 1\leq i \leq n,\ A_{i}\neq 0\}|$.
The weight distribution can give the minimum distance of the code. Moreover,
it allows the computation of the error probability of error detection and correction \cite{16K11}.

The weight distribution of linear codes is an important research topic in coding theory.
Recently, some researches devoted themselves to calculating weight distributions of linear codes \cite{5CK12,9DL16,DLM11,19LF08,23YY17}.
Linear codes of a few weights can be applied to secret sharing \cite{21YD06},
association schemes \cite{4CG84}, combinatorial designs \cite{O17}, authentication codes \cite{8DH07} and strongly regular graphs \cite{5CK86}.
There are some studies about linear codes with a few weights, for which the reader can refer to \cite{JL19,16KY19,26TX17,LL18,LY17,LC17} and reference therein.

The weight hierarchy of linear codes is another important research topic in coding theory \cite{1BL14,3CC97,11HP98,9HK92,14JL97,20WJ91,27WZ94,22YL15}.
We recall the definition of the generalized Hamming weights of linear codes \cite{20WJ91}. For an $[n,k,d]$ code $C$ and $1\leq r\leq k$,
denote by
$ [C,r]_{p} $ the set of all its $\mathbb{F}_{p}$-vector subspaces with dimension $r$.
For $ H \in [C,r]_{p}$, define $ \textrm{Supp}(H)=\cup_{c\in V}\textrm{Supp}(c)$, where $\textrm{Supp}(c)$ is the set of coordinates where $c$ is nonzero, that is,
$$ \textrm{Supp}(H)=\{i:1\leq i\leq n, c_i\neq 0 \ \ \textrm{for some $c=(c_{1}, c_{2}, \cdots , c_{n})\in H$}\}.$$
The $r$-th generalized Hamming weight (GHW) $ d_{r}(C)$ of $C$ is defined to be
$$
d_{r}(C)=\min\{|\Supp(H)|:H\in [C,r]_{p}\}, \ 1\leq r\leq k.
$$
It is easy to see that $ d_{1}(C)$ is the minimum distance $d$.
The weight hierarchy of $C$ is defined as the sequence $ (d_{1}(C),d_{2}(C),\cdots,d_{k}(C))$. About more detail, one can refer to \cite{12HP03}.
In recent years, there are some results about weight hierarchy of linear code, see \cite{2B19,13JF17,18LF17,19LW19,21XL16} and reference therein.

Let $\Tr$ denote the trace function from $\F_q$ onto $\F_p$ throughout this paper.
 For $ D= \{d_{1},d_{2},\cdots,d_{n}\}\subseteq \mathbb{F}_{q}^{\ast}$, A $p$-ary linear code $C_{D}$ of length $n$ is defined by
\begin{eqnarray}\label{defcode0}
         C_{D}=\{\left( \mathrm{Tr}(xd_1), \mathrm{Tr}(xd_2),\ldots, \mathrm{Tr}(xd_{n})\right):x\in \mathbb{F}_{p^{e}}\}.
\end{eqnarray}
Here $D$ is called the defining set of $C_{D}$.
It is a generic construction of linear codes,
which was proposed by Ding et al. \cite{5DJ15,6DD14}. Using this constructing method, many linear
codes with few weights were constructed \cite{25DD15,9DL16,19TXF17,23YY17,24ZL16}.

Motivated by the above construction, Li et al. \cite{LBY19} defined a p-ary linear code as follows. Recall that the ordinary inner product of vectors $\x=(x_1,x_2,\cdots,x_s),\ \y=(y_1,y_2,\cdots,y_s) \in \mathbb{F}_{q}^{s}$ is
$$
\x\cdot \y=x_{1}y_{1}+x_{2}y_{2}+\cdots+x_{s}y_{s}.
$$
A $p$-ary linear code $ C_{\D} $ with length $ n $ can be defined by
\begin{eqnarray}\label{defcode1}
C_{\D}=\{\left( \mathrm{Tr}(\x\cdot \d_1), \mathrm{Tr}(\x\cdot \d_2),\cdots, \mathrm{Tr}(\x\cdot \d_{n})\right):\x\in \mathbb{F}_{q}^{s}\},
\end{eqnarray}
where $ \D= \{\d_{1},\d_{2},\cdots,\d_{n}\}\subseteq \mathbb{F}_{q}^{s}\backslash\{(0,0,\cdots,0)\}$ is also called the defining set of $ C_{\D}$.
Using this method, some classes of linear codes with few weights have been constructed \cite{1AK17,JL19,LL20-1,LL20,LY17}.

Tang et al. \cite{19TXF17} constructed a $p$-ary linear code $C_D$ of \eqref{defcode0} with at most five nonzero weights from inhomogeneous quadratic function and their defining set is $D=\{x\in\F_q^*: f(x)-\Tr(\alpha x)=0\}$, where
$\alpha\in\mathbb{F}_{q}^{\ast}$ and $f(x)$ is a homogeneous quadratic function from $\F_q$ onto $\F_p$ defined by
\begin{eqnarray}\label{eq:f}
        f(x)=\sum_{i=0}^{e-1} \mathrm{Tr}(a_{i}x^{p^{i}+1}), \ a_{i}\in \mathbb{F}_{q}.
\end{eqnarray}


In this paper, inspired by the works of \cite{19LF17,LBY19,19TXF17}, we choose a defining set contained in $\mathbb{F}_{q}^{2} $ as follows.
For $\alpha\in\mathbb{F}_{q}^{\ast}$, set
\begin{align*}
\D=\D_{\alpha}&=\{(x,y)\in \mathbb{F}_{q}^{2}\setminus\{(0,0)\}: f(x)+\mathrm{Tr}(\alpha y)=0\}\\
&=\{\d_1,\d_2,\ldots,\d_n\},
\end{align*}
where $f(x)$ is defined in \eqref{eq:f} and non-degenerate.
So, the corresponding $p$-ary linear codes $C_{\D}$ of \eqref{defcode1} is
\begin{equation}\label{defcode}
C_{\D}=\Big\{\left( \Tr(\x\cdot \d_1), \Tr(\x\cdot \d_2),\ldots, \Tr(\x\cdot \d_{n})\right):\x\in \mathbb{F}_{q}^{2}\Big\}.
\end{equation}
We mainly determine their weight distributions and weight hierarchies.

The rest of the paper is arranged as follows.
Sec. 2  introduces some basic notations and results about quadratic forms useful in subsequent sections.
Sec. 3 shows the presented linear codes with three weights and determines their weight distributions and weight hierarchies.
Sec. 4 summarizes the paper.

\section{Preliminaries}\label{Pre}

In this section, we state some notations and basic facts on quadratic forms and $f$ defined in \eqref{eq:f}. These results will be used in the rest of the paper.

\subsection{Some notations fixed throughout this paper}
For Convenience, we fix the following notations. For some results on cyclotomic field $ \mathbb{Q}(\zeta_{p}) $, one can refer to \cite{IR90}.
\begin{itemize}
\item Let $\mathrm{Tr}$ be the trace function from $\mathbb{F}_{q}$ to $\mathbb{F}_{p}$.
Namely, for each $x\in \mathbb{F}_{q}$,
$$
\mathrm{Tr}(x)=x+x^{p}+ \cdots +x^{p^{e-1}}.
$$
\item $p^{\ast}=(-1)^{\frac{p-1}{2}}p$.
\item $\zeta_{p}=\exp(\frac{2\pi i}{p})$ is the primitive $p$-root of unity.
\item $\bar{\eta}$ is the quadratic character of $\mathbb{F}_{p}^{\ast}$.
It is extended by letting $\bar{\eta}(0)=0$.
\item Let $\mathbb{Z}$ be the rational integer ring and $\mathbb{Q}$ be the rational field. Let $\mathbb{K}$ be the cyclotomic
field $\mathbb{Q}(\zeta_{p})$. The field extension $\mathbb{K}/\mathbb{Q}$ is Galois of degree $p-1$.
The Galois group $Gal(\mathbb{K}/\mathbb{Q})=\{\sigma_{z}: z\in (\mathbb{Z}/p\mathbb{Z})^{\ast}\}$, where the automorphism
$\sigma_{z}$ is defined by $\sigma_{z}(\zeta_{p})=\zeta_{p}^{z}$.
\item $\sigma_{z}(\sqrt{p^{\ast}})=\bar{\eta}(z)\sqrt{p^{\ast}}$, for $1\leq z \leq p-1$.
\item Let $\langle\alpha_{1},\alpha_{2},\cdots,\alpha_{r}\rangle$ denote a space spanned by $\alpha_{1},\alpha_{2},\cdots,\alpha_{r}$.
\end{itemize}

\subsection{Quadratic form}

View $\mathbb{F}_{q}$ as a linear space over $\mathbb{F}_{p}$. Fixing $\upsilon_{1},\upsilon_{2},\cdots,\upsilon_{e} \in \mathbb{F}_{q} $ as its basis.
For any $x=x_{1}\upsilon_{1}+x_{2}\upsilon_{2}+\cdots+x_{e}\upsilon_{e} \in \mathbb{F}_{q}$ with $x_i\in\F_p, i=1,2,\cdots,e$,
there is a $\F_p$-linear isomorphism $\F_q\simeq\F_p^e$ defined as: $x=x_{1}\upsilon_{1}+x_{2}\upsilon_{2}+\cdots+x_{e}\upsilon_{e} \mapsto X=(x_1,x_2,\cdots,x_e)$,
where $X=(x_1,x_2,\cdots,x_e)$ is called the coordinate vector of $x$ under the basis $v_1,v_2,\cdots,v_e$ of $\F_q$.
A quadratic form $g$ over $\mathbb{F}_{q}$ with values in $\mathbb{F}_{p}$ can be represented by
$$
g(x)=g(X)=g(x_{1},x_{2},\cdots,x_{e})=\sum_{1\leq i,j\leq e}a_{ij}x_{i}x_{j},
$$
where $a_{ij}\in \mathbb{F}_{p}$ and $a_{ij}=a_{ji}$. We then associate with $g$ the $e\times e$ matrix $A$ whose $(i,j)$ entry is $a_{i,j}$. Then $g$ is given by $XAX^T$, where $X^T$ is the transposition of $X$.
Obviously, the matrix $A$ is symmetric and $R_g=\textrm{Rank}\ A$ called the rank of $g$, there exits an invertible matrix $M$ over $\F_p$ such that
$$
MAM^T=\textrm{diag}(\lambda_{1},\lambda_{2},\cdots,\lambda_{R_{g}},0,\cdots,0)
$$
is a diagonal matrix, where $\lambda_{1},\lambda_{2},\cdots,\lambda_{R_{g}}\in\F_p^*$.
Let $\Delta_{g}=\lambda_{1}\lambda_{2}\cdots\lambda_{R_{g}}$, and $\Delta_{g}=1$ if $R_{g}=0.$
We call $\bar{\eta}(\Delta_{g})$ the sign $\varepsilon_{g}$ of the quadratic form $g$.
It is an invariant under nonsingular linear transformations in matrix.

Let
\begin{equation}\label{eq:F}
    F(x,y)=\frac{1}{2}[g(x+y)-g(x)-g(y)].
\end{equation}
For an $r$-dimensional subspace $H$ of $\F_q$, its dual space $H^{\perp_g}$ is defined by
$$
H^{\perp_g}=\{x\in \mathbb{F}_{q}:\ F(x,y)=0, \ for \ any \  y \in H\}.
$$
Restricting the quadratic form $g$ to $H$,
it becomes a quadratic form denoted by $g|_{H}$ over $H$ in $r$ variables.
In this situation, we denote by $R_{H}$ and $\varepsilon_{H}$ the rank and sign of $g|_{H}$, respectively.

From now on, we suppose $R_{g}=e$, i.e., $g$ is a non-degenerate quadratic form. For $a\in \mathbb{F}_{p}$, set
\begin{eqnarray}\label{defcode3}
        \overline{D}_{a}=\{x\in \mathbb{F}_{q}|g(x)=a\}.
\end{eqnarray}
In the following sequel, we shall give some lemmas, which are essential in proving our main
results.

\begin{lemma}[{\cite[Proposition 1]{19LF17}}]\label{lem:1}
Let $g$ be a non-degenerate quadratic form and $H$ be a $r$-dimensional nonzero subspace of $\mathbb{F}_{q}$, then
$$
|H\cap \overline{D}_{a}|=\left\{\begin{array}{ll}
p^{r-1}+v(a)\overline{\eta}((-1)^{\frac{R_{H}}{2}})\varepsilon_{H}p^{r-\frac{R_{H}+2}{2}},  &\textrm{if\ } \ R_{H}\equiv0\pmod2, \\
p^{r-1}+\overline{\eta}((-1)^{\frac{R_{H}-1}{2}}a)\varepsilon_{H}p^{r-\frac{R_{H}+1}{2}},  &\textrm{if\ } \ R_{H}\equiv1\pmod2,
\end{array}
\right.
$$
where $v(a)=p-1$ if $a=0$, otherwise $v(a)=-1$.
\end{lemma}

\begin{lemma}[{\cite[Proposition 2]{19LF17}}]\label{lem:2}
Let $g$ be a non-degenerate quadratic form. For each $ r\ (0<2r < e)$, there exist an $r$-dimensional subspace
$H\subseteq \mathbb{F}_{p^{e}}(e>2)$ such that $H\subseteq H^{\perp_g}$.
\end{lemma}

\begin{lemma}[{\cite[Proposition 3]{19LF17}}]\label{lem:3}
Let $g$ be a non-degenerate quadratic form and $ e=2s>2$. There exists an $s$-dimensional subspace $H_{s}\subset \mathbb{F}_{p^{e}}$
such that $H_{s}=H^{\perp_g}_{s}$ if and only if $\varepsilon_{g}=(-1)^{\frac{e(p-1)}{4}}$.
\end{lemma}

\begin{lemma}[{\cite[Theorem 1]{19LF17}}]\label{lem:4}
Let $g$ be a non-degenerate quadratic form and $a$ be a non-element in $\F_p$. If $e(e > 2)$ is even, then for the linear codes $C_{\overline{D}_{a}}$ in \eqref{defcode0}
 with defining sets $\overline{D}_{a}$ defined in \eqref{defcode3}, we have
$$
d_{r}(C_{\overline{D}_{a}})=\left\{\begin{array}{ll}
p^{e-1}-p^{e-r-1}-((-1)^{\frac{e(p-1)}{4}}\varepsilon_{g}+1)p^{\frac{e-2}{2}}, & \textrm{if\ } \ 1\leq r\leq \frac{e}{2}, \\
p^{e-1}-2p^{e-r-1}-(-1)^{\frac{e(p-1)}{4}}\varepsilon_{g}p^{\frac{e-2}{2}}, & \textrm{if\ } \ \frac{e}{2}\leq r< e, \\
p^{e-1}-(-1)^{\frac{e(p-1)}{4}}\varepsilon_{g}p^{\frac{e-2}{2}}, & \textrm{if\ } \  r= e.
\end{array}
\right.
$$
\end{lemma}

\begin{lemma}[{\cite[Theorem 2]{19LF17}}]\label{lem:5}
Let $g$ be a non-degenerate quadratic form and $a$ be a non-element in $\F_p$. If $\bar{\eta}(a)=(-1)^{\frac{(e-1)(p-1)}{4}}\varepsilon_{g}$ and $e\ (e\geq3)$ is odd, then
for the linear codes $C_{\overline{D}_{a}}$ in \eqref{defcode0} with defining sets $\overline{D}_{a}$ defined in \eqref{defcode3}, we have
$$
d_{r}(C_{\overline{D}_{a}})=\left\{\begin{array}{ll}
p^{e-1}-p^{e-r-1}, & \textrm{if\ } \ 1\leq r< \frac{e}{2}, \\
p^{e-1}+p^{\frac{e-1}{2}}-2p^{e-r-1}, & \textrm{if\ } \ \frac{e}{2}< r< e, \\
p^{e-1}+p^{\frac{e-1}{2}}, & \textrm{if\ } \ r= e.
\end{array}
\right.
$$
\end{lemma}

\subsection{Some auxiliary results}
In this following sequel, we present some auxiliary results about $f$ defined in \eqref{eq:f}, which will play important roles in solving the weight distributions and weight hierarchies.
For more detail, one can refer to \cite{19TXF17}.

 For any $x\in\F_q$, $x$ can be uniquely expressed as $x=x_{1}\upsilon_{1}+x_{2}\upsilon_{2}+\cdots+x_{e}\upsilon_{e}$ with $x_i\in\F_p$.
Hence, we have
\begin{align*}
        f(x)&=\sum_{i=0}^{e-1} \mathrm{Tr}(a_{i}x^{p^{i}+1})=\sum_{i=0}^{e-1} \Tr\Big(a_{i}\Big(\sum_{j=1}^{e}x_j\upsilon_j\Big)^{p^{i}+1}\Big)\\
&=\sum_{i=0}^{e-1} \Tr\Big(a_{i}\Big(\sum_{j=1}^{e}x_j\upsilon_j^{p^i}\Big)\Big(\sum_{k=1}^{e}x_k\upsilon_k\Big)\Big)\\
&=\sum_{j=1}^{e}\sum_{k=1}^{e}\Big(\sum_{i=0}^{e-1}\Tr(a_iv_j^{p^i}v_k)\Big)x_jx_k=XBX^T,
\end{align*}
where $X=(x_1,x_2,\cdots,x_e)$ is the coordinate vector of $x$ under the basis $v_1,v_2,\cdots,v_e$ of $\F_q$ and $B=\Big(\frac{1}{2}\sum_{i=0}^{e-1}(\Tr(a_i(v_j^{p^i}v_k+v_jv_k^{p^i}))\Big)_{e\times e}$. Thus, $f$ is a quadratic form. For any $x,y\in\F_q$, by definition of $F(x,y)$ in \eqref{eq:F}, we have
\begin{align*}
F(x,y)&=\frac{1}{2}[f(x+y)-f(x)-f(y)]=\frac{1}{2}\sum_{i=0}^{e-1}\Tr(a_i(x^{p^i}y+xy^{p^i})\\
&=\frac{1}{2}\sum_{i=0}^{e-1}\Tr(a_ix^{p^i}y)+\frac{1}{2}\sum_{i=0}^{e-1}\Tr((a_{i}^{p^{-i}}x^{p^{-i}}y)^{p^i})\\
&=\frac{1}{2}\sum_{i=0}^{e-1}\Tr(a_ix^{p^i}y)+\frac{1}{2}\sum_{i=0}^{e-1}\Tr(a_{i}^{p^{e-i}}x^{p^{e-i}}y)\\
&=\Tr\Big(\Big(a_0x+\frac{1}{2}\sum_{i=1}^{e-1}(a_i+a_{e-i}^{p^i})x^{p^i}\Big)y\Big)\\
&=\mathrm{Tr}(yL_{f}(x)),
\end{align*}
where $L_{f}$ is a linearized polynomial over $\mathbb{F}_q$
defined as
\begin{eqnarray}\label{eq:L_f}
        L_{f}(x)=a_0x+\frac{1}{2}\sum_{i=1}^{e-1}(a_i+a_{e-i}^{p^i})x^{p^i}.
\end{eqnarray}
Let $\mathrm{Im}(L_{f})=\{L_{f}(x):x\in\F_q\},\ \Ker(L_{f})=\{x\in\F_q:L_{f}(x)=0\}$ denote the image and kernel of $L_{f}$, respectively.
If $L_{f}(a)=-\frac{b}{2},$ we denote $a$ by $x_{b}$.

The following two lemmas  are essential in proving our main results.

\begin{lemma}[{\cite[Lemma 5]{19TXF17}}]\label{lem:6}
 Let the symbols and notations be as above and $f$ defined in \eqref{eq:f} and $b\in \mathbb{F}_{q}$. Then
\begin{enumerate}
\item \ \ $\sum\limits_{x\in \mathbb{F}_{q}}\zeta_{p}^{f(x)}=\varepsilon_{f}q(p^{\ast})^{-\frac{R_{f}}{2}}$.
\item \ \
$
\sum\limits_{x\in \mathbb{F}_{q}}\zeta_{p}^{f(x)-\mathrm{Tr}(bx)}
=\left\{\begin{array}{ll}
0, & \textrm{if\ } b\notin \mathrm{Im}(L_{f}), \\
\varepsilon_{f}q(p^{\ast})^{-\frac{R_{f}}{2}}\zeta_{p}^{-f(x_{b})},  & \textrm{if\ } b\in \mathrm{Im}(L_{f}).
\end{array}
\right.
$
\end{enumerate}
where $x_{b}$ satisfies $L_{f}(x_{b})=-\frac{b}{2}$.
\end{lemma}

\begin{lemma}[{\cite[Lemma 4]{19TXF17}}]\label{lem:7}
 With the symbols and notations above, we have the following.
\begin{enumerate}
\item \ \ $
\sum\limits_{y\in \mathbb{F}_{p}^{\ast}}\sigma_{y}((p^{\ast})^{-\frac{r}{2}})
=\left\{\begin{array}{ll}
0, & \textrm{if\ } r  \ \textrm{is odd\ }, \\
(p^{\ast})^{-\frac{r}{2}}(p-1),  & \textrm{if\ } r \ \textrm{is even\ }.
\end{array}
\right.
$
\item \ \ For any $z\in \mathbb{F}_{p}^{\ast}$, then
$$
\sum\limits_{y\in \mathbb{F}_{p}^{\ast}}\sigma_{y}((p^{\ast})^{-\frac{r}{2}}\zeta_{p}^{z})
=\left\{\begin{array}{ll}
\bar{\eta}(z)(p^{\ast})^{-\frac{r-1}{2}}, & \textrm{if\ } r  \ \textrm{is odd\ }, \\
-(p^{\ast})^{-\frac{r}{2}},  & \textrm{if\ } r \ \textrm{is even\ }.
\end{array}
\right.
$$
\end{enumerate}
\end{lemma}
\section{Two classes of linear codes}

In this section, we study the weight distribution and weight hierarchy of linear code $C_\D$ in \eqref{defcode}, where its defining set is
\begin{equation}\label{eq:def-set}
\D=\D_{\alpha}=\{(x,y)\in \mathbb{F}_{q}^{2}\setminus\{(0,0)\}: f(x)+\Tr(\alpha y)=0\},
\end{equation}
with $\alpha\in\mathbb{F}_{q}^{\ast}$ and $f(x)$ is defined in \eqref{eq:f}.
In this paper, we suppose that $f(x)$ is a non-degenerate quadratic form. So, we have $R_f = e, \Ker(L_{f}) = \{0\}$ and $\Im(L_{f}) = \F_q$.

\subsection{Some parameters of the presented linear code}

In this subsection, we calculate the length of $C_{\D}$ defined in \eqref{defcode} and the Hamming weight of non-zero codewords of $ C_{\D}$.

\begin{lemma}\label{lem:length}
Let $\alpha\in \mathbb{F}_{q}^{\ast}$ and $C_{\D}$ be defined in \eqref{defcode}. Define $n=|\D|$. Then,
$$n=p^{2e-1}-1.$$
\end{lemma}
\begin{proof}
By the orthogonal property of additive characters, we have
\begin{align*}
n
&=\frac{1}{p}\sum_{x,y\in \mathbb{F}_{q}}\sum_{z\in \mathbb{F}_{p}}\zeta_{p}^{z(f(x)+\mathrm{Tr}(\alpha y))}-1  \\
&=\frac{1}{p}\sum_{x,y\in \mathbb{F}_{q}}(1+\sum_{z\in \mathbb{F}_{p}^{\ast}}\zeta_{p}^{z(f(x)+\mathrm{Tr}(\alpha y))})-1  \\
&=\frac{1}{p}q^{2}+\frac{1}{p}\sum_{z\in \mathbb{F}_{p}^{\ast}}\sum_{x,y\in \mathbb{F}_{q}}\zeta_{p}^{z(f(x)+\mathrm{Tr}(\alpha y))}-1  \\
&=\frac{1}{p}q^{2}+\frac{1}{p}\sum_{z\in \mathbb{F}_{p}^{\ast}}\sum_{y\in \mathbb{F}_{q}}\zeta_{p}^{\mathrm{Tr}(z\alpha y)}\sum_{x\in \mathbb{F}_{q}}\zeta_{p}^{zf(x)}-1.
\end{align*}
Since $\alpha\neq0$, we have $ \sum_{y\in \mathbb{F}_{q}}\zeta_{p}^{\mathrm{Tr}(z\alpha y)}=0$. So
$n=\frac{1}{p}q^{2}-1=p^{2e-1}-1$. The proof is finished.
\end{proof}

 For any $(u,v)\in\F_q^2\setminus\{(0,0)\}$, let $c_{(u,v)}$ be the corresponding codeword in $C_{\D}$ defined in \eqref{defcode}, that is,
$$
c_{(u,v)}=\Big(\Tr(ux+vy)\Big)_{(x,y)\in \D}.
$$

\begin{lemma}\label{lem:wt}
Let $(u,v)(\neq (0,0))\in\F_{q}^{2}$. We have
\begin{enumerate}
\item If $v\in \mathbb{F}_{q}\setminus \mathbb{F}_{p}^*\alpha$, then $\wt(c_{(u,v)})=p^{2e-2}(p-1)$.
\item If $v\in \mathbb{F}_{p}^{\ast}\alpha$ and $u=0$, then
$$
\wt(c_{(u,v)})
=\left\{\begin{array}{ll}
p^{2e-2}(p-1), & \textrm{if\ } e  \ \textrm{is odd\ }, \\
p^{2e-2}(p-1)\Big(1-\varepsilon_{f}(p^{\ast})^{-\frac{e}{2}}\Big),  & \textrm{if\ } e  \ \textrm{is even\ }.
\end{array}
\right.
$$
\item If $v\in \mathbb{F}_{p}^{\ast}\alpha, \ u\neq0$, and $f(x_{u})=0$, then
$$
\wt(c_{(u,v)})
=\left\{\begin{array}{ll}
p^{2e-2}(p-1), & \textrm{if\ } e  \ \textrm{is odd\ }, \\
p^{2e-2}(p-1)\Big(1-\varepsilon_{f}(p^{\ast})^{-\frac{e}{2}}\Big),  & \textrm{if\ } e  \ \textrm{is even\ }.
\end{array}
\right.
$$
\item If $v\in \mathbb{F}_{p}^{\ast}\alpha, \ u\neq0$, and $f(x_{u})\neq0$, then
$$
\wt(c_{(u,v)})
=\left\{\begin{array}{ll}
p^{2e-2}\Big(p-1-\varepsilon_{f}\bar{\eta}(-f(x_{u}))(p^{\ast})^{-\frac{e-1}{2}}\Big), & \textrm{if\ } e  \ \textrm{is odd\ }, \\
p^{2e-2}\Big(p-1+\varepsilon_{f}(p^{\ast})^{-\frac{e}{2}}\Big),  & \textrm{if\ } e  \ \textrm{is even\ }.
\end{array}
\right.
$$
\end{enumerate}
\end{lemma}
\begin{proof}
Put
$N(u,v)=\{(x,y)\in \F_{q}^{2}: f(x)+\Tr(\alpha y)=0, \Tr(ux+vy) = 0\}$, then $\wt(c_{(u,v)})=n-|N(u,v)|+1$.
So, we only need to calculate $|N(u,v)|$.

By the orthogonal property of additive characters, we have
\begin{align*}
&|N(u,v)|=\frac{1}{p^{2}}\sum_{x,y\in \mathbb{F}_{q}}\Big(\sum_{z_{1}\in \mathbb{F}_{p}}\zeta_{p}^{z_{1}f(x)+\mathrm{Tr}(z_{1}\alpha y)}\sum_{z_{2}\in \mathbb{F}_{p}}\zeta_{p}^{\mathrm{Tr}(z_{2}(ux+vy))}\Big)  \\
&=\frac{1}{p^{2}}\sum_{x,y\in \mathbb{F}_{q}}\Big((1+\sum_{z_{1}\in \mathbb{F}_{p}^{\ast}}\zeta_{p}^{z_{1}f(x)+\mathrm{Tr}(z_{1}\alpha y)})(1+\sum_{z_{2}\in \mathbb{F}_{p}^{\ast}}\zeta_{p}^{\mathrm{Tr}(z_{2}(ux+vy))})\Big)  \\
&=\frac{1}{p^{2}}q^{2}+\frac{1}{p^{2}}\sum_{z_{1}\in \mathbb{F}_{p}^{\ast}}\sum_{x,y\in \mathbb{F}_{q}}\zeta_{p}^{z_{1}f(x)+\mathrm{Tr}(z_{1}\alpha y)}+ \frac{1}{p^{2}}\sum_{z_{2}\in \mathbb{F}_{p}^{\ast}}\sum_{x,y\in \mathbb{F}_{q}}\zeta_{p}^{\mathrm{Tr}(z_{2}(ux+vy))} \\
&+\frac{1}{p^{2}}\sum_{z_{1}\in \mathbb{F}_{p}^{\ast}}\sum_{z_{2}\in \mathbb{F}_{p}^{\ast}}\sum_{x,y\in \mathbb{F}_{q}}\zeta_{p}^{z_{1}f(x)+\mathrm{Tr}(z_{2}ux+z_{2}vy+z_{1}\alpha y)}  \\
&=p^{2e-2}+p^{-2}\sum_{z_{1}\in \mathbb{F}_{p}^{\ast}}\sum_{z_{2}\in \mathbb{F}_{p}^{\ast}}\sum_{y\in \mathbb{F}_{q}}\zeta_{p}^{\mathrm{Tr}(z_{2}vy+z_{1}\alpha y)}\sum_{x\in \mathbb{F}_{q}}\zeta_{p}^{z_{1}f(x)+\mathrm{Tr}(z_{2}ux)}.
\end{align*}

If $v\in \mathbb{F}_{q}\setminus \mathbb{F}_{p}^{\ast}\alpha$, then $\sum_{y\in \mathbb{F}_{q}}\zeta_{p}^{\mathrm{Tr}(z_{2}vy+z_{1}\alpha y)}=0$.
Hence, the result of (1) holds directly.

If $v\in \mathbb{F}_{p}^{\ast}\alpha$, i.e., $\alpha=zv$ for some $z\in \mathbb{F}_{p}^{\ast}$, then
\begin{align}
|N(u,v)|
&=p^{2e-2}+p^{e-2}\sum_{z_{1}\in \mathbb{F}_{p}^{\ast}}\sum_{x\in \mathbb{F}_{q}}\zeta_{p}^{z_{1}f(x)-z_{1}\mathrm{Tr}(zux)}. \nonumber
\end{align}

If $v\in \mathbb{F}_{p}^{\ast}\alpha$ and $u=0$, then, by Lemma~\ref{lem:6} and Lemma~\ref{lem:7}, we have
\begin{align*}
|N(u,v)|
&=p^{2e-2}+p^{e-2}\sum_{z_{1}\in \mathbb{F}_{p}^{\ast}}\sum_{x\in \mathbb{F}_{q}}\zeta_{p}^{z_{1}f(x)}  \\
&=p^{2e-2}+p^{e-2}\sum_{z_{1}\in \mathbb{F}_{p}^{\ast}}\sigma_{z_{1}}(\varepsilon_{f}q(p^{\ast})^{-\frac{R_{f}}{2}})  \\
&=p^{2e-2}+p^{2e-2}\varepsilon_{f}\sum_{z_{1}\in \mathbb{F}_{p}^{\ast}}\sigma_{z_{1}}((p^{\ast})^{-\frac{e}{2}})  \\
&=\left\{\begin{array}{ll}
p^{2e-2}, & \textrm{if\ } e  \ \textrm{is odd\ }, \\
p^{2e-2}\Big(1+\varepsilon_{f}(p^{\ast})^{-\frac{e}{2}}(p-1)\Big),  & \textrm{if\ } e  \ \textrm{is even\ }.
\end{array}
\right.
\end{align*}

If $v\in \mathbb{F}_{p}^{\ast}\alpha, u\neq0$, then
\begin{align*}
|N(u,v)|
&=p^{2e-2}+p^{e-2}\sum_{z_{1}\in \mathbb{F}_{p}^{\ast}}\sigma_{z_{1}}\Big(\sum_{x\in \mathbb{F}_{q}}\zeta_{p}^{f(x)-\mathrm{Tr}(zux)}\Big).
\end{align*}
Let $c=zu$. So $x_{c}=zx_{u}$. By Lemma~\ref{lem:6}, we have
\begin{align}
|N(u,v)|
&=p^{2e-2}+p^{e-2}\sum_{z_{1}\in \mathbb{F}_{p}^{\ast}}\sigma_{z_{1}}\Big(\varepsilon_{f}q(p^{\ast})^{-\frac{R_{f}}{2}}\zeta_{p}^{-f(x_{c})}\Big) \nonumber \\
&=p^{2e-2}+p^{2e-2}\varepsilon_{f}\sum_{z_{1}\in \mathbb{F}_{p}^{\ast}}\sigma_{z_{1}}\Big((p^{\ast})^{-\frac{e}{2}}\zeta_{p}^{-f(x_{c})}\Big) \nonumber \\
&=p^{2e-2}+p^{2e-2}\varepsilon_{f}\sum_{z_{1}\in \mathbb{F}_{p}^{\ast}}\sigma_{z_{1}}\Big((p^{\ast})^{-\frac{e}{2}}\zeta_{p}^{-z^{2}f(x_{u})}\Big). \nonumber
\end{align}
The last two results follow directly from lemma~\ref{lem:7}. The proof is finished.
\end{proof}

\begin{remark}\label{rem:dim}
By Lemma \ref{lem:wt}, we know that, for $(u,v)(\neq (0,0))\in\F_{q}^{2}$, we have $\wt(\c_{(u,v)})>0$. So, the map: $\F_q^2\rightarrow C_{\D}$ defined by $(u,v)\mapsto \c_{(u,v)} $
is an isomorphism as linear spaces over $\F_p$. Hence, the dimension of the codes $C_{\D}$ in \eqref{defcode} is equal to $2e$.
\end{remark}

\begin{lemma}\label{lem:dis}
Let $C_{\D}$ be defined in \eqref{defcode}. Then, the minimal distance of the dual code $C_{\D}^{\perp}$ is at least $2$.
\end{lemma}
\begin{proof}
 We prove it by contradiction. If not, then there exists a coordinate $i$ such that the $i$-th entry of all of the codewords of $C_{\D}$ is $0$,
that is, $\Tr(ux+vy)=0$ for all $(u,v)\in \F_{q}^{2}$, where $(x,y)\in \D$.
Thus, by the properties of the trace function, we have $(x,y)=(0,0)$. It contradicts with $(x,y) \neq (0,0)$.
\end{proof}

\subsection{Weight distribution}

In this subsection, we give the weight distributions of $C_{\D}$ in \eqref{defcode} case by case on the term of the parity of $e$.

\begin{theorem}\label{thm:wd} Let the code $C_{\D}$ be defined in \eqref{defcode}. Then,
\begin{itemize}
\item[(1)] when $e$ is odd, the code $C_{\D}$ is a $[p^{2e-1}-1,2e]$ linear code over $ \mathbb{F}_{p}  $
with the weight distribution in Table 1.
\begin{table}\label{tab:wd:o}
\centering
\caption{The weight distribution of $C_{\D}$ of Theorem \ref{thm:wd}}
\begin{tabular*}{11.5cm}{@{\extracolsep{\fill}}ll}
\hline
\textrm{Weight} $\omega$ & \textrm{Multiplicity} $A_\omega$   \\
\hline
0 &   1  \\
$(p-1)p^{2e-2}$ &  $p^{2e}-(p-1)^{2}p^{e-1}-1$  \\
$(p-1)p^{2e-2}-p^{\frac{3e-3}{2}}$  & $\frac{1}{2}(p-1)^{2}(p^{e-1}+p^{\frac{e-1}{2}})$  \\
$(p-1)p^{2e-2}+p^{\frac{3e-3}{2}}$  & $\frac{1}{2}(p-1)^{2}(p^{e-1}-p^{\frac{e-1}{2}})$  \\
\hline
\end{tabular*}
\end{table}
\item[(2)] when $e$ is even, set $\varepsilon=(-1)^{\frac{e(p-1)}{4}}\varepsilon_{f}$. The code $C_{\D}$ is a $[p^{2e-1}-1,2e]$ linear code over $ \mathbb{F}_{p}  $
with the weight distribution in Table 2.
\begin{table}
\centering
\caption{The weight distribution of $C_{\D}$ of Theorem \ref{thm:wd}}
\begin{tabular*}{11.5cm}{@{\extracolsep{\fill}}ll}
\hline
\textrm{Weight} $\omega$ \qquad& \textrm{Multiplicity} $A_\omega$   \\
\hline
0 \qquad&   1  \\
$(p-1)p^{2e-2}$ \qquad&  $p^{2e}-p^{e}(p-1)-1$  \\
$p^{\frac{3e-4}{2}}(p-1)(p^{\frac{e}{2}}-\varepsilon)$  \qquad& $p^{\frac{e-2}{2}}\Big(p-1\Big)\Big(p^{\frac{e}{2}}+\varepsilon(p-1)\Big)$  \\
$p^{\frac{3e-4}{2}}\Big((p-1)p^{\frac{e}{2}}+\varepsilon\Big)$  \qquad& $p^{\frac{e-2}{2}}(p-1)^{2}(p^{\frac{e}{2}}-\varepsilon)$  \\
\hline
\end{tabular*}
\end{table}
\end{itemize}
\end{theorem}

\begin{proof}

(1) Assume that $(u,v)\neq (0,0)$. By Lemma \ref{lem:length} and Lemma \ref{lem:wt},
$\wt(c_{(u,v)})$ has only three values, that is,
$$
\left\{\begin{array}{ll}
\omega_{1}=(p-1)p^{2e-2}, \  \\
\omega_{2}=(p-1)p^{2e-2}-p^{\frac{3e-3}{2}}, \  \\
\omega_{3}=(p-1)p^{2e-2}+p^{\frac{3e-3}{2}}.
\end{array}
\right.
$$

Recall that $A_{\omega_{i}}$ is the multiplicity of $\omega_{i}$. By Lemma \ref{lem:wt}, we have
\begin{align*}\label{eq:1}
  &A_{\omega_{1}} = \Big|\Big\{(u,v)\in\F_q^2|\wt(c_{(u,v)}) = (p-1)p^{2e-2}\Big\}\Big|  \\
  &= \Big|\Big\{(u,v)\in\F_q^2\setminus\{(0,0)\}|u\in\F_q,v\in\F_q\setminus\F_p^*\alpha\Big\}\Big|\\
  &+\Big|\Big\{(0,v)\in\F_q^2|v\in\F_p^*\alpha\Big\}\Big|+\Big|\Big\{(u,v)\in\F_q^2|v\in\F_p^*\alpha, u\neq 0, f(x_u)=0\Big\}\Big|\\
&=q\Big(q-(p-1)\Big)-1+(p-1)+(p-1)(p^{e-1}-1)\\
&=p^{2e}-(p-1)^2p^{e-1}-1,
\end{align*}
where we use the fact that the number of non-zero solutions of the equation $f(x)=0$ in $\F_q$ is $p^{e-1}-1$, one can refer to \cite[Theorem 6.27]{16LN97}.

By Lemma \ref{lem:dis} and the first two Pless Power Moment (\cite[P. 259]{12HP03} ),
we obtain the system of linear equations as follows:
$$
\left\{\begin{array}{ll}
A_{\omega_{1}}=p^{2e}-(p-1)^{2}p^{e-1}-1 \  \\
A_{\omega_{2}}+A_{\omega_{3}}=(p-1)^{2}p^{e-1} \  \\
\omega_{1}A_{\omega_{1}}+\omega_{2}A_{\omega_{2}}+\omega_{3}A_{\omega_{3}}=p^{2e-1}(p^{2e-1}-1)(p-1).
\end{array}
\right.
$$
Solving the system , we get
$$
\left\{\begin{array}{ll}
A_{\omega_{1}}=p^{2e}-(p-1)^{2}p^{e-1}-1 \  \\
A_{\omega_{2}}=\frac{1}{2}(p-1)^{2}(p^{e-1}+p^{\frac{e-1}{2}}) \  \\
A_{\omega_{3}}=\frac{1}{2}(p-1)^{2}(p^{e-1}-p^{\frac{e-1}{2}}).
\end{array}
\right.
$$
Thus,
we obtain the Table 1.

(2) As the proof of (1), assume that $(u,v)\neq (0,0)$. By Lemma \ref{lem:length} and Lemma \ref{lem:wt},
$\wt(c_{(u,v)})$ has only three values, that is,
$$
\left\{\begin{array}{ll}
\omega_{1}=(p-1)p^{2e-2}, \  \\
\omega_{2}=p^{\frac{3e-4}{2}}(p-1)(p^{\frac{e}{2}}-\varepsilon), \  \\
\omega_{3}=p^{\frac{3e-4}{2}}\Big((p-1)p^{\frac{e}{2}}+\varepsilon\Big).
\end{array}
\right.
$$

By Lemma \ref{lem:wt} and Lemma \ref{lem:1} (or \cite[Theorem 6.26]{16LN97}), we have
$$
\left\{\begin{array}{ll}
A_{\omega_{1}}=p^{2e}-(p-1)p^{e}-1, \  \\
A_{\omega_{2}}=p^{\frac{e-2}{2}}\Big(p-1\Big)\Big(p^{\frac{e}{2}}+\varepsilon(p-1)\Big), \  \\
A_{\omega_{3}}=p^{\frac{e-2}{2}}(p-1)^{2}(p^{\frac{e}{2}}-\varepsilon).
\end{array}
\right.
$$
Thus,
we obtain the Table 2.
\end{proof}

\begin{example}
Let $(p,e,\alpha)=(5,3,1)$ and $f(x)=\mathrm{Tr}(x^{2})$.
Then, the corresponding code $C_{\D_{1}}$ has parameters $[ 3124,6,2375]$ and weight enumerator
$1+240x^{2375}+15224x^{2500}+160x^{2625}$.
\end{example}

\begin{example}
Let $(p,e,\alpha)=(3,4,1)$ and $f(x)=\mathrm{Tr}(\theta x^{2})$. By Corollary 1 in \cite{19TXF17}, we have
$\varepsilon_{f}=1$. Then, the corresponding code $C_{\D}$ has parameters $[2186,8,1296]$ and weight enumerator
$1+66x^{1296}+6398x^{1458}+96x^{1539}$. Here $\theta$ is a primitive element of $\mathbb{F}_{q}$.
\end{example}

\begin{example}
Let $(p,e,\alpha)=(3,4,\theta)$ and $f(x)=\mathrm{Tr}(x^{2})$. By Corollary 1 in \cite{19TXF17}, we have
$\varepsilon_{f}=-1$. Then, the corresponding code $C_{\D}$ has parameters $[2186,8,1377]$ and weight enumerator
$1+120x^{1377}+6398x^{1458}+42x^{1620}$. Here $\theta$ is a primitive element of $\mathbb{F}_{q}$.
\end{example}

\subsection{Weight hierarchy}

In this subsection, we give the weight hierarchies of $C_{\D}$ in \eqref{defcode} case by case on the term of the parity of $e$.

By Remark~\ref{rem:dim}, we know that the dimension of the code $C_{\D}$ defined in \eqref{defcode} is $2e$.
So, by Proposition 2.1 \cite{LL20}, we give a general formula,
that is
\begin{align}
        d_{r}(C_{\D})&=n-\max\{|H_r^\perp\cap \D|: H_r \in [\mathbb{F}_{q}^{2},r]_{p}\} \\
        &=n-\max\{|H_{2e-r}\cap \D|: H_{2e-r} \in [\mathbb{F}_{q}^{2},2e-r]_{p}\}\label{eq:d_r:2},
\end{align}
which will be employed to calculate the generalized Hamming weight $d_r(C_\D)$. Here $H_r^\perp = \{\y\in\F_q^2:\Tr(\x\cdot \y) =0, \textrm{for any $\y\in H_r$} \}$.

Let $H_r$ be an $r$-dimensional subspace of $\mathbb{F}_{q}^{2}$ and $\beta_{1},\cdots,\beta_{r}$ be an $\mathbb{F}_{p}$-basis of $H_r$. Set
$$
N(H_r)=\{\x=(x,y)\in \mathbb{F}_{q}^{2}: f(x)+\mathrm{Tr}(\alpha y)=0, \mathrm{Tr}(\x\cdot \beta_{i})=0, 1\leq i\leq r\}.
$$
Then, $N(H_r) = (\D\cap H_r^\perp)\cup\{(0,0)\}$, which concludes that $|N(H_r)|=|H_r^\perp\cap\D |+1$. Hence, we have
\begin{equation}\label{eq:d_r:3}
     d_{r}(C_{\D})=n+1-\max\{N(H_r)|: H_r \in [\mathbb{F}_{q}^{2},r]_{p}\}
\end{equation}

\begin{lemma} \label{lem:d_r:2}
Let $H_r$ and $N(H_r)$ be defined as above. We have
\begin{itemize}
  \item[(1)] if $\alpha\notin  \Prj_{2}(H_r)$, $N(H_r)=p^{2e-(r+1)}$.
  \item[(2)] if $\alpha\in  \Prj_{2}(H_r)$,
  $$N(H_r)=p^{2e-(r+1)}\Big(1+\varepsilon_{f}\sum\limits_{(y_1,-u)\in H_r}\sum\limits_{z\in \F_p^{\ast}}\sigma_{z}\big((p^{\ast})^{-\frac{e}{2}}\zeta_{p}^{f(x_{y_1})}\big)\Big).$$
\end{itemize}
Here $ \Prj_{2}$ is the second projection from $\mathbb{F}_{q}^{2}$ to $\mathbb{F}_{q}$ defined by $(x,y)\mapsto y$.
\end{lemma}

\begin{proof}
By the orthogonal property of additive characters, we have
\begin{align*}
&p^{r+1}|N(H_r)|
=\sum_{\x=(x,y)\in \F_q^2}\sum_{z\in \F_p}\zeta_{p}^{z f(x)+\Tr(z\alpha y)}\prod_{i=1}^{r}\sum_{x_{i}\in \F_p}\zeta_{p}^{\Tr(x_{i}(\x\cdot \beta_{i}))}  \\
&=\sum_{\x=(x,y)\in \F_q^{2}}\sum_{z\in \F_p}\zeta_{p}^{z f(x)+\Tr(z(\alpha y))}\sum_{\y\in H_r}\zeta_{p}^{\Tr(\x\cdot\y)} \\
&=\sum_{\x=(x,y)\in \F_q^{2}}\sum_{\y\in H_r}\zeta_{p}^{\Tr(\x\cdot \y)}
+\sum_{\x=(x,y)\in \F_q^{2}}\sum_{z\in \F_p^{\ast}}\sum_{\y\in H_r}\zeta_{p}^{z f(x)+\Tr(\x\cdot\y+z \alpha y)} \\
&=q^{2}+\sum_{\x=(x,y)\in \F_q^{2}}\sum_{z\in \F_p^{\ast}}\sum_{\y\in H_r}\zeta_{p}^{z f(x)+\Tr(\x\cdot\y+z \alpha y)}.
\end{align*}
where the last equation comes from
\begin{align*}
\sum_{\x=(x,y)\in \F_{q}^{2}}\sum_{\y\in H_{r}}(-1)^{\Tr(\y\cdot \x)}
&=\sum_{\x=(x,y)\in \F_{q}^{2}} \ 1+\sum_{\x=(x,y)\in \F_{q}^{2}}\sum_{(0,0)\neq\y\in H_{r}}(-1)^{\Tr(\y\cdot \x)}  \\
&=q^{2}+\sum_{(0,0)\neq\y\in H_{r}}\sum_{\x=(x,y)\in \F_{q}^{2}}(-1)^{\Tr(\y\cdot \x)}=q^{2}.
\end{align*}

Denote $B_{H_r}=\sum_{\x=(x,y)\in \F_q^{2}}\sum_{z\in \F_p^{\ast}}\sum_{\y\in H_r}\zeta_{p}^{z f(x)+\Tr(\x\cdot\y+z \alpha y)}$ for short.
Then
$$p^{r+1}|N(H_r)| = q^2+B_{H_r}$$
 and
\begin{align*}
    B_{H_r} &=\sum_{\x=(x,y)\in \F_q^{2}}\sum_{z\in \F_p^{\ast}}\sum_{\y\in H_r}\zeta_{p}^{z f(x)+\Tr(\x\cdot\y+z \alpha y)} \\
&=\sum_{(x,y)\in \F_q^{2}}\sum_{z\in \F_p^{\ast}}\sum_{(y_1,y_2)\in H_r}\zeta_{p}^{z f(x)+\Tr(y_1x+y_2y+z \alpha y)}   \\
&=\sum_{(y_1,y_2)\in H_r}\sum_{z\in \F_p^{\ast}}\sum_{x\in \F_q}\zeta_{p}^{z f(x)+\Tr(y_1x)}\sum_{y\in \F_q}\zeta_{p}^{\Tr(y_2y+z \alpha y)}    \\
&=\sum_{(y_1,y_2)\in H_r}\sum_{z\in \F_p^{\ast}}\sum_{x\in \F_q}\zeta_{p}^{z f(x)+z\Tr(\frac{y_1}{z}x)}\sum_{y\in \F_q}\zeta_{p}^{z\Tr(\frac{y_2}{z}y+\alpha y)}   \\
&=\sum_{(y_1,y_2)\in H_r}\sum_{z\in \F_p^{\ast}}\sum_{x\in \F_q}\zeta_{p}^{z f(x)+z\Tr(y_1x)}\sum_{y\in \F_q}\zeta_{p}^{z\Tr(y_2y+\alpha y)}
\end{align*}

If $\alpha\notin  \Prj_{2}(H_r)$, then $\sum_{y\in \mathbb{F}_{q}}\zeta_{p}^{z\Tr(y_2y+\alpha y)}=0$. So $B_{H_r}=0$.

If $\alpha\in  \Prj_{2}(H_r)$, by Lemma~\ref{lem:6}, we have
\begin{align*}
B_{H_r}&=q\sum_{(y_1,-\alpha)\in H_r}\sum_{z\in \F_p^{\ast}}\sum_{x\in \F_{q}}\zeta_{p}^{z f(x)+z\Tr(y_1x)}    \\
&=q\sum_{(y_1,-\alpha)\in H_r}\sum_{z\in \F_p^{\ast}}\sigma_{z}\Big(\sum_{x\in \F_q}\zeta_{p}^{f(x)+\Tr(y_1x)}\Big)    \\
&=q\sum_{(y_1,-\alpha)\in H_r}\sum_{z\in \F_p^{\ast}}\sigma_{z}\Big(\varepsilon_{f}q(p^{\ast})^{-\frac{R_{f}}{2}}\zeta_{p}^{f(x_{y_1})}\Big)    \\
&=\varepsilon_{f}q^{2}\sum_{(y_1,-\alpha)\in H_r}\sum_{z\in \F_p^{\ast}}\sigma_{z}\Big((p^{\ast})^{-\frac{e}{2}}\zeta_{p}^{f(x_{y_1})}\Big)
\end{align*}

So, the desired result is obtained. Thus, we complete the proof.
\end{proof}

In the following sequel, we shall determine the weight hierarchy of $C_{\D}$ in \eqref{defcode}  by calculating $N(H_r)$ in Lemma~\ref{lem:d_r:2} and $|\D\cap H_{2e-r}|$ in \eqref{eq:d_r:2}.

\begin{theorem}\label{thm:wh}
 Let the code $C_{\D}$ be defined in \eqref{defcode} and $e\geq 3$. Define
$$
e_0=\left\{\begin{array}{ll}
\frac{e-1}{2}, & \textrm{if $e$ is odd}, \\
\frac{e}{2}, & \textrm{if $e$ is even and $\varepsilon_{f}=(-1)^{\frac{e(p-1)}{4}}$},\\
\frac{e-2}{2}, & \textrm{if $e$ is even and $\varepsilon_{f}=-(-1)^{\frac{e(p-1)}{4}}$}.
\end{array}
\right.
$$ Then,
\begin{itemize}
\item[(1)]
when $e-e_0+1 \leq r \leq 2e$, we have
$$
d_{r}(C_{\D})=p^{2e-1}-p^{2e-r}.
$$
\item[(2)] when $0\leq r\leq e-e_0$, we have
$$
d_{r}(C_{\D})=\left\{\begin{array}{ll}
p^{2e-1}-p^{2e-r-1}-p^{\frac{3e-3}{2}},  \textrm{if $2\nmid e$\ }, \\
p^{2e-1}-p^{2e-r-1}-(p-1)p^{\frac{3e-4}{2}},  \textrm{if $2\mid e$, $\varepsilon_{f}=(-1)^{\frac{e(p-1)}{4}}$},\\
p^{2e-1}-p^{2e-r-1}-p^{\frac{3e-4}{2}},  \textrm{if\ } \ \textrm{if $2\mid e$, $\varepsilon_{f}=-(-1)^{\frac{e(p-1)}{4}}$}.
\end{array}
\right.
$$
\end{itemize}
\end{theorem}

\begin{proof}
(1) when $e-e_0+1 \leq r \leq 2e$, then $0 \leq 2e-r \leq e_0+e-1$.
 Let $T_{\alpha}=\{x\in \mathbb{F}_{q}:\ \mathrm{Tr}(\alpha x)=0\}$.
It is easy to know that $\dim(T_{\alpha})=e-1$.
By Lemma~\ref{lem:2} and Lemma~\ref{lem:3}, there exists an $e_0$-dimensional subspace $J_{e_0}$ of $\F_q$ such that $f(x)=0$ for any $x\in J_{e_0}$.
Note that the dimension of the subspace $J_{e_0}\times T_{\alpha}$ is $e_0+e-1$. Let $H_{2e-r}$ be a $(2e-r)$-dimensional subspace of $J_{e_0}\times T_{\alpha}$,
then,
$$
|H_{2e-r}\cap \D|=2^{2e-r}-1,
$$
Hence, by \eqref{eq:d_r:2}, we have
$$
d_{r}(C_{\D})=n-\max\{|\D \cap H|: H \in [\mathbb{F}_{p^{e}}^{2},2e-r]_{p}\}=p^{2e-1}-p^{2e-r}.
$$

Thus, it remains to determine $d_r(C_{\D})$ when $0 < r \leq e-e_0$.

(2) When $0\leq r\leq e-e_0$, we discuss case by case.

\textbf{Case 1:}\ $e (e\geq 3)$ is odd. In this case, $e_0 = \frac{e-1}{2}$ and $e-e_0 = \frac{e+1}{2}$, that is, $0\leq r\leq \frac{e+1}{2}$.

Suppose $H_r$ is an $r$-dimensional subspace of $\F_q^2$. If $\alpha\in  \Prj_{2}(H_{r})$, by Lemma~\ref{lem:7} and Lemma~\ref{lem:d_r:2}, we have
\begin{align*}
N(H_r)=p^{2e-(r+1)}\Big(1+\varepsilon_{f}(-1)^{\frac{(e-1)(p-1)}{4}}p^{-\frac{e-1}{2}}\sum_{(y_1,-\alpha)\in H_{r}}\bar{\eta}(f(x_{y_1})\big)\Big).
\end{align*}

Now we want to construct $H_r$ such that $N(H_r)$ reaches its maximum, that is, the number of such as $(y_1,-\alpha)$ is maximal in $H_{r}$ and for any $(y_1,-\alpha)\in H_{r}, \bar{\eta}(f(x_{y_1})) = \varepsilon_{f}(-1)^{\frac{(e-1)(p-1)}{4}}$. The constructing method is as follows.

Taking an element $a\in\F_p^*$ satisfying $\overline{\eta}(a)=(-1)^{\frac{(e-1)(p-1)}{4}}\epsilon_{f}$, then, by Lemma~\ref{lem:1} (or \cite[Theorem 6.27]{16LN97}), we know that the length and the dimension of $C_{\overline{D}_{a}}$ in \eqref{defcode0} are $p^{e-1}+p^{\frac{e-1}{2}}$ and $e$, respectively.
Combining formula \eqref{eq:d_r:2} with Lemma~\ref{lem:5}, we have
\begin{align*}
    d_{e-r}(C_{\overline{D}_{a}})&=p^{e-1}+p^{\frac{e-1}{2}}-\max\{|\overline{D}_{a} \cap H|: H \in [\F_q,r]_{p}\}\\
    &=p^{e-1}+p^{\frac{e-1}{2}}-2p^{e-(e-r)-1},
\end{align*}
which follows that
$$
\max\{|\overline{D}_{a} \cap H|: H \in [\F_q,r]_{p}\}=2p^{r-1}.
$$
Thus, there exists an $r$-dimensional subspace $J_{r}$ of $\F_q$ such that $|\overline{D}_{a} \cap J_r| = 2p^{r-1}$. By Lemma~\ref{lem:1}, we know that $R_{J_{r}}=1$ and $\epsilon_{J_r}=\overline{\eta}(a)$, which concludes that there exists an $(r-1)$-dimensional subspace $J_{r-1}$ of $J_{r}$ satisfying $f(J_{r-1})=0$ and
$\overline{\eta}(f(x))=(-1)^{\frac{(e-1)(p-1)}{4}}\epsilon_{f}$, for each $x\in J_{r}\setminus J_{r-1}$.
Let $\alpha_{1},\alpha_{2},\cdots,\alpha_{r-1}$ be an $\mathbb{F}_{p}$-basis of $J_{r-1}$. Take an element $\alpha_{r}\in J_{r}\setminus J_{r-1}$ and
set
$$
\mu_{1}=\alpha_{1}+\alpha_{r},\mu_{2}=\alpha_{2}+\alpha_{r},\cdots,\mu_{r-1}=\alpha_{r-1}+\alpha_{r},\mu_{r}=\alpha_{r}.
$$
Then $\mu_{1},\mu_{2},\cdots,\mu_{r-1},\mu_{r}$ is an $\mathbb{F}_{p}$-basis of $J_{r}$.
Define
$$
\lambda_{1}=(\mu_{1},-\alpha),\lambda_{2}=(\mu_{2},-\alpha),\cdots,\lambda_{r-1}=(\mu_{r-1},-\alpha),\lambda_{r}=(\mu_{r},-\alpha)
$$
and $V_{r}=\langle\lambda_{1},\lambda_{2},\cdots,\lambda_{r-1},\lambda_{r}\rangle$.
Then $V_{r}$ is an $r$-dimensional subspace of $\mathbb{F}_{q}^{2}$.
Set $S(-\alpha)=\{(y,z)\in V_{r}:\ z=-\alpha\}$. It is easy to know that the cardinal number of $S(-\alpha)$ is $p^{r-1}$.
We assert that $f(y)=f(\alpha_{r})$ for any $(y,-\alpha)\in S(-\alpha)$. In fact,
note that $(y,-\alpha)$ has the following unique representation:
\begin{align*}
   (y,-\alpha)&=x_{1}\lambda_{1}+x_{2}\lambda_{2}+\cdots+x_{r-1}\lambda_{r-1}+x_{r}\lambda_{r},\\
   &=(x_{1}\alpha_{1}+x_{2}\alpha_{2}+\cdots+x_{r-1}\alpha_{r-1}+\alpha_{r},-\alpha), x_{i}\in \mathbb{F}_{p},
\end{align*}
then, we have $y=x_{1}\alpha_{1}+x_{2}\alpha_{2}+\cdots+x_{r-1}\alpha_{r-1}+\alpha_{r}$.
Thus, we have $f(y)=f(\alpha_{r})$.

Take $H_{r}$ as follows:
$$
H_{r}=\langle(L_f(\mu_{1}),-\alpha),(L_f(\mu_{2}),-\alpha),\cdots,(L_f(\mu_{r-1}),-\alpha),(L_f(\mu_{r}),-\alpha)\rangle,
$$
then $H_{r}$ is our desired $r$-dimensional subspace of $\F_q^2$ and its $N(H_r)$ reaches the maximum
$$ N(H_r)=p^{2e-(r+1)}\Big(1+p^{r-1-\frac{e-1}{2}}\Big) = p^{2e-r-1} + p^{\frac{3e-3}{2}}. $$
So, for $0\leq r\leq\frac{e+1}{2}$, the desired result is obtained by Lemma~\ref{lem:d_r:2} and \eqref{eq:d_r:3}.

\textbf{Case 2:}\ $e(e\geq 3) $ is even and $\varepsilon_{f}=(-1)^{\frac{e(p-1)}{4}}$. In this case, $e_0 = \frac{e}{2}$ and $e-e_0 = \frac{e}{2}$, that is, $0\leq r\leq \frac{e}{2}$.

Suppose $H_r$ is an $r$-dimensional subspace of $\F_q^2$ and $\alpha\in  \Prj_{2}(H_{r})$. Recall that $v(0)=p-1$ and $v(x)=-1$ for $x\in\mathbb{F}_{p}^{\ast}$ defined in Lemma~\ref{lem:1}.
By Lemma~\ref{lem:7} and Lemma~\ref{lem:d_r:2}, we have
\begin{align*}
N(H_r)&=p^{2e-(r+1)}\Big(1+\varepsilon_{f}\sum\limits_{(y_1,-u)\in H_r}\sum\limits_{z\in \F_p^{\ast}}\sigma_{z}\big((p^{\ast})^{-\frac{e}{2}}\zeta_{p}^{f(x_{y_1})}\big)\Big)\\
&=p^{2e-(r+1)}\Big(1+p^{-\frac{e}{2}}\sum_{(y_1,-\alpha)\in H_{r}}v(f(x_{y_1})\Big).
\end{align*}

Let $J_{r}$ be a subspace of $J_{\frac{e}{2}}$ with an basis $\mu_{1},\mu_{2},\cdots,\mu_{r}$.
Take $H_{r}$ as follows:
$$
H_{r}=\langle(L_f(\mu_{1}),-\alpha),(L_f(\mu_{2}),-\alpha),\cdots,(L_f(\mu_{r-1}),-\alpha),(L_f(\mu_{r}),-\alpha)\rangle,
$$
then  $N(H_r)$ reaches its maximum
$$ N(H_r)=p^{2e-(r+1)}\Big(1+(p-1)p^{r-1-\frac{e}{2}}\Big) = p^{2e-r-1} + (p-1)p^{\frac{3e-4}{2}}. $$
So, for $0\leq r\leq\frac{e}{2}$, the desired result is obtained by Lemma~\ref{lem:d_r:2} and \eqref{eq:d_r:3}.

\textbf{Case 3}:\ $e (e\geq 3)$ is even and $\varepsilon_{f}=-(-1)^{\frac{e(p-1)}{4}}$. In this case, $e_0 = \frac{e-2}{2}$ and $e-e_0 = \frac{e}{2}+1$, that is, $0\leq r\leq \frac{e}{2}+1$.

Suppose $H_r$ is an $r$-dimensional subspace of $\F_q^2$ and $\alpha\in  \Prj_{2}(H_{r})$.
By Lemma~\ref{lem:7} and Lemma~\ref{lem:d_r:2}, we have
\begin{align}
N(H_r)=p^{2e-(r+1)}\Big(1-p^{-\frac{e}{2}}\sum_{(y_1,-\alpha)\in H_{r}}v(f(x_{y_1})\Big).   \nonumber
\end{align}

For $1\leq r\leq\frac{e}{2}$, taking an element $a\in\F_p^*$, then, by Lemma~\ref{lem:1} (or \cite[Theorem 6.26]{16LN97}), we know that the length and the dimension of $C_{\overline{D}_{a}}$ in \eqref{defcode0} are $p^{e-1}+p^{\frac{e-2}{2}}$ and $e$, respectively.
Combining formula \eqref{eq:d_r:2} with Lemma~\ref{lem:4}, we have
\begin{align*}
    d_{e-r}(C_{\overline{D}_{a}})&=p^{e-1}+p^{\frac{e-2}{2}}-\max\{|\overline{D}_{a} \cap H|: H \in [\F_q,r]_{p}\}\\
    &=p^{e-1}-2p^{e-(e-r)-1}+p^{\frac{e-2}{2}},
\end{align*}
which follows that
$$
\max\{|\overline{D}_{a} \cap H|: H \in [\F_q,r]_{p}\}=2p^{r-1}.
$$
Thus, there exists an $r$-dimensional subspace $J_{r}$ of $\F_q$ such that $|\overline{D}_{a} \cap J_r| = 2p^{r-1}$. By Lemma~\ref{lem:1}, we know that $R_{J_{r}}=1$ and $\epsilon_{J_r}=\overline{\eta}(a)$, which concludes that there exists an $(r-1)$-dimensional subspace $J_{r-1}$ of $J_{r}$ satisfying $f(J_{r-1})=0$ and
$\overline{\eta}(f(x))=(-1)^{\frac{(e-1)(p-1)}{4}}\epsilon_{f}$, for each $x\in J_{r}\setminus J_{r-1}$.
Let $\alpha_{1},\alpha_{2},\cdots,\alpha_{r-1}$ be an $\mathbb{F}_{p}$-basis of $J_{r-1}$. Take an element $\alpha_{r}\in J_{r}\setminus J_{r-1}$ and
set
$$
\mu_{1}=\alpha_{1}+\alpha_{r},\mu_{2}=\alpha_{2}+\alpha_{r},\cdots,\mu_{r-1}=\alpha_{r-1}+\alpha_{r},\mu_{r}=\alpha_{r}.
$$
Then $\mu_{1},\mu_{2},\cdots,\mu_{r-1},\mu_{r}$ is an $\mathbb{F}_{p}$-basis of $J_{r}$. Take $H_{r}$ as follows:
$$
H_{r}=\langle(L_f(\mu_{1}),-\alpha),(L_f(\mu_{2}),-\alpha),\cdots,(L_f(\mu_{r-1}),-\alpha),(L_f(\mu_{r}),-\alpha)\rangle,
$$
then $N(H_r)$ reaches its maximum

$$ N(H_r)=p^{2e-(r+1)}\Big(1+p^{r-1-\frac{e}{2}}\Big) = p^{2e-r-1} + p^{\frac{3e-4}{2}}. $$
So, , for $0\leq r\leq\frac{e}{2}$, the desired result is obtained by Lemma~\ref{lem:d_r:2} and \eqref{eq:d_r:3}.

For $r=\frac{e}{2}+1$, taking an element $a\in\F_p^*$, then by Lemma~\ref{lem:4}, we have
\begin{align*}
    d_{e-r}(C_{\overline{D}_{a}})&=p^{e-1}+p^{\frac{e-2}{2}}-\max\{|\overline{D}_{a} \cap H|: H \in [\F_q,r]_{p}\}\\
    &=p^{e-1}-p^{e-(e-r)-1},
\end{align*}
which follows that
$$
\max\{|\overline{D}_{a} \cap H|: H \in [\mathbb{F}_{q},\frac{e}{2}+1]_{p}\}=p^{\frac{e}{2}}+p^{\frac{e-2}{2}}.
$$
Thus, there exists an $r$-dimensional subspace $J_{r}$ of $\F_q$ such that $|\overline{D}_{a} \cap J_r| = p^{\frac{e-2}{2}}+p^{\frac{e}{2}}$.
By the proof of Lemma~\ref{lem:4} (or \cite[Theorem 1]{19LF17}), we know that $R_{J_r}=2$ and $\epsilon_{J_r}=\bar{\eta}((-1)^{\frac{e}{2}-1})\epsilon_f$.
So, there exists a $(r-2)$-dimensional subspace $J_{r-2}$ of $ J_{r}$ satisfying $f(J_{r-2})=0$.
Let $\alpha_{1},\alpha_{2},\cdots,\alpha_{r-2}$ be an $\mathbb{F}_{p}$-basis of $J_{r-2}$.
Choose two elements $\gamma_{1}, \gamma_{2}\in J_{r}\setminus J_{r-2}$
such that $\alpha_{1},\cdots,\alpha_{r-2},\gamma_{1}, \gamma_{2}$ is an $\mathbb{F}_{p}$-basis of $J_{r}$.
Set
$$
\mu_{1}=\alpha_{1}+\gamma_{2},\cdots,\mu_{r-2}=\alpha_{r-2}+\gamma_{2},\mu_{r-1}=\gamma_{1}+\gamma_{2},\mu_{r}=\gamma_{2}.
$$
Then $\mu_{1},\mu_{2},\cdots,\mu_{r-1},\mu_{r}$ is an $\mathbb{F}_{p}$-basis of $J_{r}$. Take $H_{r}$ as follows:
$$
H_{r}=\langle(L_f(\mu_{1}),-\alpha),(L_f(\mu_{2}),-\alpha),\cdots,(L_f(\mu_{r-1}),-\alpha),(L_f(\mu_{r}),-\alpha)\rangle,
$$
then $N(H_r)$ reaches its maximum

$$ N(H_r)=p^{2e-(r+1)}\Big(1+p^{r-1-\frac{e}{2}}\Big) = 2p^{\frac{3e-4}{2}}=p^{2e-r-1} + p^{\frac{3e-4}{2}}. $$
So, for $ r=\frac{e}{2}+1$, the desired result is obtained.

Thus, we complete the proof.
\end{proof}

\section{Concluding Remarks}

In this paper, we constructed a family of $3$-weight linear codes and determined their weight distributions and weight hierarchies.
Here, we note that, in our case, the quadratic function $f$ defined in \eqref{eq:f} is non-degenerate.
If $f$ is degenerate, one can refer to \cite{LL20-0,19TXF17} and obtain similar results by using our strategy.

Let $w_{\min}$ and $w_{\max}$ denote the minimum and maximum nonzero weight of the linear code $C_{\D_{\alpha}}$, respectively.
If the code $e\geq3$, then it can be easily checked that
$$
 \frac{w_{\min}}{w_{\max}}> \frac{p-1}{p}.
$$
By the results in \cite{21YD06}, most of the codes we construct
are suitable for constructing secret sharing schemes with interesting properties.

\section*{Acknowledgements}

The research was supported by Anhui Provincial Natural Science Foundation No. 1908085MA02 and the National Science Foundation of China Grant No.12001312 and No. 11701001.




\end{document}